\documentclass[11pt]{article}

\usepackage[a4paper,margin=1.1in]{geometry}
\usepackage{lmodern}  
\usepackage[T1]{fontenc}
\usepackage[utf8]{inputenc}
\usepackage{amsmath}
\usepackage{amssymb}
\usepackage{mathrsfs}
\usepackage{amsthm}
\usepackage{microtype}
\usepackage{listings}
\usepackage[hidelinks]{hyperref}

\lstdefinestyle{lean}{
  basicstyle=\small\ttfamily,
  columns=fullflexible,
  keepspaces=true,
  showstringspaces=false,
  breaklines=false,
  xleftmargin=1em,
  literate=
    {¬}{{$\lnot$}}1
    {∣}{{$\mid$}}1
    {≤}{{$\leq$}}1
    {ℝ}{{$\mathbb{R}$}}1
    {ℕ}{{$\mathbb{N}$}}1
    {σ}{{$\sigma$}}1
    {φ}{{$\varphi$}}1
    {Φ}{{$\Phi$}}1
}
\lstnewenvironment{lean}{\lstset{style=lean}}{}

\theoremstyle{plain}
\newtheorem{theorem}{Theorem}[section]
\newtheorem{proposition}[theorem]{Proposition}
\newtheorem{lemma}[theorem]{Lemma}

\theoremstyle{definition}
\newtheorem{definition}[theorem]{Definition}
\newtheorem{remark}[theorem]{Remark}

\newcommand{\ZN}{\mathbb{Z}_N}
\newcommand{\phis}{\varphi_s}
\newcommand{\phid}{\varphi_d}
\newcommand{\phir}{\varphi_r}
\newcommand{\phie}{\varphi_e}
\newcommand{\phic}{\varphi_c}
\newcommand{\Rring}{\mathscr{A}_N}

\title{Existence of an infinite family of substrates satisfying all the
postulates of integrated information theory, exclusion included,
with arbitrarily large integrated information}
\author{Arnaud Mayeux}
\date{}

\begin{document}
\maketitle

\begin{abstract}
Integrated information theory assigns to a physical system in a state a number
$\Phi$, and declares a system to be a substrate of experience only if it prevails
over all of its subsystems in system integrated information. Upper bounds on $\Phi$ are
known and grow hyper-exponentially in the number of units. On the lower side there has
been nothing. We exhibit a family of substrates whose $\Phi$ is proved to grow, and which
is proved to be a complex. For every $N \ge 8$ with $3 \nmid N$, let
$\Rring$ be the cycle of $N$ binary units in which each unit takes the exclusive-or of
itself and its two clockwise neighbours. We prove two things about it, concerning two
different quantities. First, in every state and against every background, $\Rring$ is a
complex: its system integrated information $\phis$, the quantity in which candidate
substrates compete, strictly exceeds that of every nonempty proper subsystem. Second, at
the all-ones state $\mathbf{1}$, its integrated information $\Phi$---the sum over its
distinctions and their relations---satisfies
\[
  \Phi_{\max}(\Rring, \mathbf{1}) \;\ge\;
  2^{\lfloor N^{2}/8\rfloor}\cdot\frac{2^{-N}}{2N} - 1 ,
\]
so that $\log_2 \Phi$ grows quadratically in the number of units: already at twenty
units the bound exceeds $2.6\times10^{7}$, and at forty units $1.8\times10^{46}$. Both
statements, and every lemma they rest on, are verified in the Lean~4 proof assistant
against a formalization of the definitions of IIT~4.0, with no unproved steps. The bound
comes from the contiguous arcs of the cycle: we prove that every arc of length at least
three is an irreducible mechanism, and that quadratically many of them share a unit.
\end{abstract}

\section{Introduction}
\label{sec:intro}

Integrated information theory begins with experience and works backwards. From a
small set of properties that any experience is taken to have---that it exists, and
is structured, specific, unitary and definite---the theory derives conditions that a
physical system must satisfy to be the substrate of that experience, and assigns to
such a system a single number: its integrated information, a measure of how much the
system exists as one thing rather than as a collection of parts. That number is the
subject of this paper.

A basic question about any such quantity is its range. How large can the integrated
information of a system of a given size be, and which architectures achieve it? This
is a mathematical question, answerable without taking a position on the theory's
claims about consciousness, and its answer bears on those claims either way: a
quantity that varies by a small factor across all possible systems means something
different from one whose variation is astronomical.

One half of this question has been answered, and answered well. Zaeemzadeh and
Tononi~\cite{zaeemzadeh2024} established upper bounds on the integrated information
of distinctions and of their relations, and hence on that of a structure as a whole.
Their bounds grow hyper-exponentially in the number of units, and their analysis goes
further than the bounds themselves: it identifies when an individual mechanism
saturates its own bound, isolates the obstruction to saturating all of them at
once---mechanisms sharing parts cannot all be maximally integrated---and advances the
well-motivated conjecture that densely connected, grid-like architectures are the
natural candidates for high values. It is a careful and clarifying piece of work, and
it is the foundation on which the present paper builds. Its authors are explicit that
attainment in general remains open, and they are right to be: showing that some system
actually reaches a given magnitude is a different undertaking from bounding what any
system could reach.

That is the gap this paper addresses. To the best of our knowledge nothing whatever is
known on the lower side: no family of systems of growing size has been proved to have
integrated information growing at any rate, and no substrate has been proved to satisfy
the exclusion postulate. Both are absent for the same reason. Everything established so
far about attainment rests on evaluating one system at a time, and a value computed at a
given size, however exactly, settles nothing about how the quantity behaves as the size
grows. It is worth being precise about what \emph{is} available, because a good deal is. Zaeemzadeh and Tononi do not merely bound:
the proof of their Theorem~3 is constructive, exhibiting transition matrices on
overlapping mechanisms, and their numerical evaluation of those constructions grows
super-exponentially in the number of units. It has also long been appreciated, in
examples that circulate informally rather than in print, that highly expansive
architectures give large values of the bipartition-normalized measure used in earlier
formulations of the theory. Both are genuine sources of insight, and both leave open the
question at issue here, which is not whether large values occur but whether they occur
\emph{on a substrate the theory recognizes}: the constructions are not shown to satisfy
exclusion, and the authors of~\cite{zaeemzadeh2024} are explicit that this is outside
their scope, writing that their work is focused on composition and that their ``results
are applicable even if a candidate system does not satisfy all the essential
properties''---exclusion being one of the five they list. Until a family is exhibited that
prevails over all of its subsystems and whose integrated information is proved to grow
with size, the interval between zero and the known ceiling is not known to be occupied at
any point the theory would call a subject.

There are two good reasons why this half of the question has proved the harder one.

The first is the mathematics itself, which is unusually resistant. Integrated
information is not an average or a smooth functional of a system's parameters. It is
defined through several optimizations nested inside one another: over every way of
carving a system into parts, taking the worst case; over candidate causes and effects,
taking the best; and over every subset of the system, since the theory holds that the
subject is whichever subset wins the competition and that overlapping candidates cannot
coexist. Each ranges over a set growing exponentially with system size, and they
interact, so none can be treated in isolation. The quantity is also not
compositional---the value of a system is not assembled from the values of its
parts---and it can move discontinuously under small changes in wiring. The familiar
tools for showing that a quantity is large, such as averaging arguments or convex
relaxations, either do not apply or destroy the feature being measured. A lower bound
therefore requires exact control of all of the nested optimizations at once, and
uniformly in the size of the system.

The second reason concerns instruments rather than mathematics. The study of integrated
information has been advanced above all by computation, and by one instrument in
particular: PyPhi~\cite{pyphi}, the open-source implementation of the theory's
definitions, which has been indispensable to the field and to us. Nearly every
structural insight now in circulation, including several that guided this paper,
originated in exact evaluation of small systems. Such evaluation, however, carries an
inherent horizon at roughly a dozen units, and it answers questions of the form
\emph{what is the value of this system} rather than \emph{how does the value grow with
size}. There is also a subtlety that no amount of computing power removes: at the sizes
that can be evaluated, architectures with very different asymptotic behaviour may be
nearly indistinguishable, because the structures responsible for rapid growth have not
yet had room to appear. Evidence gathered within the computable range is therefore
excellent for discovering candidates and indispensable for testing conjectures, but it
cannot by itself settle a question about growth. That calls for proof---and, given the
case analyses involved, for proof of a kind impractical to carry out and to check by
hand.

Our contribution is a lower bound. We exhibit an explicit family of
architectures---uniform, sparse, and without hub or hierarchy, each unit reading three
consecutive neighbours on a cycle---and prove that its integrated information grows as
two raised to a power quadratic in the number of units. We prove that these systems
prevail over every one of their subsystems, so that the quantity bounded is the
integrated information of a genuine complex and not of a substrate the theory would
decline to recognize. The bound rests on a structural fact we prove about these systems:
every contiguous arc of the cycle of length at least three is an irreducible mechanism,
and quadratically many such arcs share any given unit. Taken with the general ceilings
already known, this places the first certified point inside the range of the theory's
central quantity.

Our method is to leverage Lean~\cite{lean4}, a proof assistant: definitions and theorems
are written in a formal language, and a small kernel checks every inference. Nothing is
granted because it looks right or because an expert vouches for it. Its library Mathlib
now runs to over a million lines, and major theorems are routinely checked this way.
What that buys here is a change of instrument rather than an increase in power:
mechanical \emph{reasoning} in place of mechanical \emph{computation}. The field's existing machinery evaluates the theory's definitions on
a given system, and is thereby confined to the systems small enough to evaluate; a proof
assistant instead manipulates those same definitions symbolically, with a system size
left as a variable, so that a single argument settles every size at once. The
bookkeeping that makes such an argument prohibitive by hand---hundreds of cases, each
individually routine---is exactly what a machine sustains reliably, and the kernel
certifies that none of it was waved through. Every theorem stated in this paper has been
verified in Lean against a formalization of the theory's definitions, with no unproved
steps and no appeals to numerical evaluation that the kernel does not itself check,
with the single exception of Proposition~\ref{prop:onlyarcs}, whose status is stated
explicitly where it appears; the development is available as~\cite{circ-repo} and can be
re-checked independently. The proofs were produced
by human mathematical direction working together with artificial-intelligence assistants
inside that formal setting. The strategic content---which architectures could satisfy
both demands at once, why hub-like designs are structurally disqualified, the reduction
of the system-level competition to a finite combinatorial question about damage and
cuts, and the identification of the irreducible parts---was ours; the assistants
supplied scale, and the proof assistant certainty. We describe this arrangement in
Section~\ref{sec:method}, as we believe it applies well beyond the present problem.

We should add that formal verification was not a precaution taken after the fact. Our
original conjecture for this family was more ambitious: we expected its integrated
information to grow doubly exponentially, and a substantial development was built toward
that end. The conjecture is false, and the argument refuting it is short---the
irreducible parts are all contiguous arcs, of which there are only quadratically many.
No formalization could have caught this, since the error lay in an expectation about the
object rather than in any inference. What the formal development did was make the
surviving results dependable, and the failure precise enough to convert into the
structure theorem reported here. It illustrates the difficulty we are describing: the
principal hazard here is not incorrect reasoning, but correct reasoning about an object
that cannot support the intended conclusion, and often only a proof carried to the end
distinguishes the two.

\paragraph{Organization.} Section~\ref{sec:prelim} fixes the definitions of IIT~4.0 in
the form in which they are formalized. Section~\ref{sec:substrate} introduces the cycle
$\Rring$ and states the two main theorems. Section~\ref{sec:damage} develops the device
on which the system-level analysis rests, an exact reduction of $\phis$ to a
combinatorial \emph{damage} count. Section~\ref{sec:exclusion} proves that $\Rring$ is a
complex. Section~\ref{sec:arcs} determines its distinctions, and
Section~\ref{sec:lower} derives the lower bound. Section~\ref{sec:method} describes the
working arrangement and Section~\ref{sec:formal} the formal development.

\paragraph{How to read this paper.} The paper is addressed to three audiences at once,
and they need not read the same parts. A reader interested in what the result says about
the theory will find the substance in Sections~\ref{sec:prelim}
and~\ref{sec:substrate}---the object, the two theorems, and what they do and do not
establish---and in Section~\ref{sec:method}; nothing there presupposes more than the
definitions of IIT~4.0, and the two theorems are stated in words before they are stated
in symbols. A reader interested in the mathematics will want
Sections~\ref{sec:damage}--\ref{sec:lower}, where the reduction of system integrated
information to a combinatorial count is set up and the classification arguments are
carried out. A reader interested in the formalization will want
Section~\ref{sec:formal}, which names the Lean theorems corresponding to each statement
and records what the kernel checks. We have kept the statements themselves exact
throughout, on the view that a claim about a quantity as intricate as this one is worth
less when it is paraphrased.

\section{Preliminaries}
\label{sec:prelim}

We assume IIT~4.0 as formulated by Albantakis et al.~\cite{iit40} and do not restate it.
The definitions are used here in their formalized form: they have been transcribed into
Lean~4 in a separate library, released as~\cite{integ-repo} and documented
in~\cite{mayeux-formalization}. That library contains the definitions together with
general results about them, among them a formalization of the ceiling
of~\cite{zaeemzadeh2024}; what it does not contain is any result about a particular
substrate. Everything proved in the present paper is stated
against that library and adds to it without modifying it, so that the correspondence
between the theory as published and the theory as formalized is settled there, once, and
is not re-opened by anything done here. We fix notation and record the three points at
which the formal setting requires a decision.

It may help to recall what the three quantities in play are for, since the argument moves
between them. The mechanism integrated information $\phid$ is local: it measures how much
a particular subset of units, in its current state, specifies a cause and an effect that
its own parts cannot account for. The system integrated information $\phis$ is global: it
measures how much a whole candidate system resists being cut into causally independent
pieces, and it is the quantity in which candidates compete, since the theory awards
substratehood to the candidate that maximizes it. Finally $\Phi$ counts the structure:
it sums $\phid$ over the irreducible mechanisms and $\phir$ over the relations among
them, and it is the quantity the theory identifies with how much an entity exists
intrinsically. The two are largely independent of one another, and this is what makes the
problem here a conjunction rather than a single optimization: we must make $\phis$ of the
whole beat $\phis$ of every part---an argument about cuts---while simultaneously making
$\Phi$ large---an argument about how many mechanisms are irreducible and how their
purviews overlap.

Throughout, a substrate is a set of $N$ binary units whose transition probability matrix
factorizes over units, and every substrate here is deterministic. For a mechanism $M$ in
state $m$ we write $\phie(m)$ and $\phic(m)$ for the mechanism integrated effect and
cause information, $\phid(m) = \min\{\phie(m),\phic(m)\}$, and, for a distinction, its
\emph{support} for the union of its maximally irreducible cause and effect purviews,
taken as sets of \emph{stated units}, that is of pairs (unit, state): a unit shared in
opposite states is not shared.
For a system partition $\theta$ of a candidate set $S$ we write
$\mathcal{N}(\theta) = \sum_{p} |p| \cdot |\mathrm{cutSet}(p)|$ for the sum over the
blocks of $\theta$ of the number of connections severed into each, and $\phis(S, s)$ for
the system integrated information of $S$ when the universe is in state $s$, obtained by minimizing
$\phis(\theta)/\mathcal{N}(\theta)$ and taking the supremum of $\phis(\theta)$ over the
minimizing set. Both $\phis$ and $\Phi$ depend on the state, and on the background $u$
against which the units outside the candidate are held; we display the state, as in
$\phis(S,s)$ and $\Phi_{\max}(T,s)$, and leave $u$ implicit, since every statement below
is either proved for all $u$ or explicitly quantified over it. The two quantities should
not be confused: $\phis$ is what candidates compete in and is the subject of the
exclusion postulate, while $\Phi$ measures the structure a winning candidate specifies.

First, ties. The maximizing cause and effect purviews of a mechanism need not be unique,
and the value of $\Phi$ can depend on which are chosen. The formalization therefore
parametrizes a cause--effect structure by a \emph{selector} $\sigma$, choosing one
congruent maximizing pair per mechanism, writes $\Phi(\sigma)$ for the resulting value,
and sets $\Phi_{\max} = \sup_\sigma \Phi(\sigma)$. Our lower bound is stated for
$\Phi_{\max}$ but is proved from a family of distinctions that every selector must
contain, so it holds for each selector individually.

Second, exclusion. A candidate $S$ is a \emph{complex} in state $s$ when
$\phis(S,s) > \phis(S',s)$ for every $S' \ne S$ overlapping $S$ within the universe
$U$~\cite[Eq.~26]{iit40}. We prove
the case $S = U$, where every nonempty candidate overlaps $S$ and the condition becomes
the statement that the whole strictly exceeds each of its nonempty proper subsystems.
This is the strongest reading: it leaves no room for a proper subsystem to tie, and none
for a disjoint competitor to escape the comparison.

Third, background. On a proper subsystem the cause side depends on the convention for
the units outside it. We use the marginalization of~\cite[Eqs.~3--4]{iit40} throughout, which
is what makes $\phis$ of a proper subsystem sensitive to the candidate, and hence what
gives the exclusion postulate content.

\section{The substrate and the results}
\label{sec:substrate}

The system we study is about as plain as a system can be. Arrange $N$ binary units in a
circle. At each tick, every unit looks at itself and its two clockwise neighbours, and
turns on exactly when an odd number of those three are on. Every unit obeys the same
rule, no unit is distinguished, and each reads only three others however large the cycle
grows. We return in Section~\ref{sec:method} to why we looked at so plain an object in
the first place.

\begin{definition}[The cycle $\Rring$]
\label{def:ring}
Fix $N \ge 3$. Index the units by the cyclic group $\ZN = \mathbb{Z}/N\mathbb{Z}$, so
that indices are read modulo $N$ throughout, and let the state space be
$\Omega_N = \{0,1\}^{\ZN}$. Let $\Rring$ be the substrate on these units whose
transition probability matrix is deterministic and given unit by unit by
\[
  p_z\bigl(s'_z = b \;\big|\; s\bigr) \;=\;
  \begin{cases}
    1, & b = s_z \oplus s_{z+1} \oplus s_{z+2},\\[2pt]
    0, & \text{otherwise,}
  \end{cases}
  \qquad z \in \ZN,\quad s \in \Omega_N,
\]
where $\oplus$ denotes addition modulo $2$. Equivalently, $\Rring$ is the deterministic
substrate whose global transition map is
\[
  F_N : \Omega_N \to \Omega_N, \qquad
  (F_N s)_z \;=\; s_z \oplus s_{z+1} \oplus s_{z+2}.
\]
The \emph{window} of unit $z$ is $W_z = \{z, z+1, z+2\}$, the set of units it reads. In
the formal development $\Rring$ is \texttt{circNet N}.
\end{definition}

Three properties of $\Rring$ are used below, and we isolate them before proceeding.

\begin{lemma}
\label{lem:ring-basic}
Let $N \ge 3$ and identify $\Omega_N$ with $\mathbb{F}_2^{\ZN}$.
\begin{enumerate}
\item[(i)] $F_N$ is $\mathbb{F}_2$-linear, with matrix the circulant
  $C = I + P + P^{2}$, where $P$ is the cyclic shift $(Ps)_z = s_{z+1}$.
\item[(ii)] $F_N$ is bijective if and only if $3 \nmid N$.
\item[(iii)] The all-ones state $\mathbf{1} \in \Omega_N$ satisfies $F_N(\mathbf{1}) =
  \mathbf{1}$.
\end{enumerate}
\end{lemma}

\begin{proof}
(i) is immediate from the definition, exclusive-or being addition in $\mathbb{F}_2$.
For (ii), circulants form an $\mathbb{F}_2$-algebra isomorphic to
$\mathbb{F}_2[x]/(x^{N}-1)$, the shift $P$ corresponding to $x$; under it $C$ becomes
$f(x) = 1 + x + x^{2}$, and $C$ is invertible exactly when $f$ is a unit, that is when
$\gcd\bigl(f(x),\, x^{N}-1\bigr) = 1$. Since
$f(x) = (x^{3}-1)/(x-1)$, the roots of $f$ are the primitive cube roots of unity, and
these are $N$-th roots of unity precisely when $3 \mid N$. For (iii), each window has
three units, and $1 \oplus 1 \oplus 1 = 1$.
\end{proof}

Item (ii) is not a technical convenience. Bijectivity is what forces the past of the cycle
to be as determinate as its future; it is the hypothesis under which the cause side of
the analysis collapses onto the effect side in Proposition~\ref{prop:damage}, and it is
the reason the arithmetic condition $3 \nmid N$ appears in both main theorems. By item
(iii) we may take $\mathbf{1}$ as the state throughout, and we do; the exclusion theorem
is nevertheless proved for every state.

\begin{remark}
Informally, $\Rring$ is uniform, sparse and flat: it is invariant under rotation, each
unit reads exactly three units however large $N$ grows, and no unit is distinguished from
any other. In particular it is the opposite of the densely connected architecture one
might expect a large value to require.
\end{remark}

Our two results are the following. The first says that the theory recognizes this cycle as
a subject rather than as a heap: no part of it, however chosen, is more integrated than
the whole. The second says that what the theory then assigns to it is astronomically
large.

\begin{theorem}[Exclusion]
\label{thm:exclusion}
For every $N \ge 8$ with $3 \nmid N$, every background $u$ and \emph{every} state $s$,
the substrate $\Rring$ is a complex:
\[
  \phis(\Rring, s) \;>\; \phis(S', s)
  \qquad \text{for every nonempty proper } S' \subsetneq \ZN .
\]
More precisely $\phis(\Rring, s) \ge 4$ while $\phis(S', s) < 4$ for every such $S'$.
\end{theorem}

\begin{theorem}[Lower bound]
\label{thm:lower}
For every $N \ge 8$ with $3 \nmid N$ and every background $u$, the integrated
information of $\Rring$ \emph{at the all-ones state} $\mathbf{1}$ satisfies
\[
  \Phi_{\max}(\Rring, \mathbf{1}) \;\ge\;
  2^{\lfloor N^2/8 \rfloor}\cdot\frac{(1/2)^{N}}{2N} \;-\; 1 .
\]
In particular $\log_2 \Phi_{\max}(\Rring, \mathbf{1}) = \Omega(N^2)$. Unlike
Theorem~\ref{thm:exclusion}, this is a statement about a single state.
\end{theorem}

Theorem~\ref{thm:lower} is stated with deliberately crude constants; the exponent
$\lfloor N^2/8\rfloor$ and the floor $2^{-N}/(2N)$ are both far from optimal, and the
displayed quantity is negative, hence vacuous, for $N \le 11$, becoming positive at
$N = 13$ and growing thereafter ($8.8$, $35.6$, $2.0\times10^3$, $2.68\times10^7$ at
$N = 13, 14, 16, 20$). We comment on the true constants in
Remark~\ref{rem:constants}. What matters is the growth rate, and for that the crude
constants suffice.

Together the two theorems say that the quantity bounded in Theorem~\ref{thm:lower} is
the integrated information of a system the theory recognizes as a subject: a complex,
not merely a set of units. To our knowledge this is the first lower bound on $\Phi$ for a
family of growing size, and the first substrate shown---rather than computed at small
sizes---to satisfy every postulate of IIT~4.0, exclusion included.

For orientation, the general ceiling of~\cite{zaeemzadeh2024}, formalized in the imported
library~\cite{integ-repo}, gives
$\Phi_{\max} = O\bigl(N^{2}\,2^{(2^{N})}\bigr)$ for every substrate on $N$ units. The
present family therefore sits strictly inside the known interval, at
$2^{\Theta(N^2)}$ rather than at the ceiling.

\section{The damage model}
\label{sec:damage}

The system level is where exclusion is decided, and the following reduction is what makes
it tractable. It replaces the nested optimizations of $\phis$ by a count.

Fix a system partition $\theta$ of a candidate set $S$. For a unit $j \in S$ let its
\emph{live inputs} be the units $z+1, z+2$ that it reads and that lie in $S$ (units
outside $S$ are held at the background and contribute nothing). Say $j$ is
\emph{damaged} by $\theta$ when some live input of $j$ lies in the cut set of $j$'s own
block, and write $\mathrm{dmg}(\theta)$ for the number of damaged units.

\begin{proposition}[Damage reduction]
\label{prop:damage}
Let $S = \ZN$ and let the state be all-ones, with $3 \nmid N$ and $N \ge 4$. Then for
every system partition $\theta$,
\[
  \phis(\theta) \;=\; \mathrm{dmg}(\theta),
\]
and consequently
\[
  \phis(\Rring, \mathbf{1}) \;=\;
  \max\Bigl\{ \mathrm{dmg}(\theta) \;:\;
  \theta \in \operatorname*{arg\,min}_{\vartheta}
  \frac{\mathrm{dmg}(\vartheta)}{\mathcal{N}(\vartheta)} \Bigr\}.
\]
\end{proposition}

The effect side of the identity is a direct computation: because the dynamics is
exclusive-or, a unit whose live input has been severed has its effect repertoire
flattened to $1/2$, while an undamaged unit keeps probability $1$; the product over units
gives $2^{-\mathrm{dmg}(\theta)}$, and the logarithm returns the count. The cause side
requires bijectivity: with $3 \nmid N$ the map is invertible, the fibre over the
all-ones state is a single point, and the cause side reduces to the effect side. This is
the point at which the arithmetic condition on $N$ enters.

On a proper subsystem the same computation gives a rescaled statement, which we record
because it is what disposes of the competitors. Let $C_S$ denote the principal submatrix
of $C$ on $S$.

\begin{proposition}[Competitors]
\label{prop:competitor}
For every nonempty proper $S \subsetneq \ZN$ at the all-ones state,
\[
  \phis(S, \mathbf{1}) \;=\; 2^{-\operatorname{corank}(C_S)}\cdot \mathrm{dmgval}(S),
\]
where $\mathrm{dmgval}(S)$ is the maximum of $\mathrm{dmg}$ over the minimizers of
$\mathrm{dmg}/\mathcal{N}$ among partitions of $S$.
\end{proposition}

The factor $2^{-\operatorname{corank}}$ is the size of the fibre of the conditioned map,
and is the only way the state enters: at corank $0$ the map is bijective and
$\phis(S,\cdot)$ is the same in every state. Two consequences are used below. First, subtractive
corrections to the damage count are \emph{false}: we record this because it is the
natural guess and it fails. Second, a singular subsystem is penalized by a factor of at
least two, which is what makes the crude bound of Lemma~\ref{lem:crude} sufficient
against a large class of candidates.

\section{Exclusion}
\label{sec:exclusion}

Theorem~\ref{thm:exclusion} has two halves: the whole cycle achieves $\phis \ge 4$, and
every proper subsystem stays strictly below $4$. Both are proved by classifying
partitions according to how much damage they do.

\paragraph{The whole cycle: $\phis \ge 4$.} By Proposition~\ref{prop:damage} it suffices
to exhibit a partition with a small ratio and to show that no partition of damage at
most three can match it. The classification is where the work lies, and it proceeds by
the damage budget.

\begin{lemma}[$D=1$]
\label{lem:d1}
A partition of $\ZN$ damaging exactly one unit has $\mathcal{N}(\theta) = N-1$.
\textup{(\texttt{sysCutCount\_eq\_of\_damagedCount\_one})}
\end{lemma}

\begin{lemma}[$D=2$]
\label{lem:d2}
A partition of $\ZN$ damaging exactly two units has either
$\mathcal{N}(\theta) = 2(N-1)$, the case of two singleton inputs blocks, or
$4\,\mathcal{N}(\theta) \le N^{2}$, the case of one inputs block and its complement.
\textup{(\texttt{sysCutCount\_of\_damagedCount\_two})}
\end{lemma}

\begin{lemma}[$D=3$]
\label{lem:d3}
For $N \ge 8$, a partition of $\ZN$ damaging exactly three units satisfies
$4\,\mathcal{N}(\theta) \le (N-1)^2 + 8(N-1)$.
\textup{(\texttt{sysCutCount\_of\_damagedCount\_three'})}
\end{lemma}

\noindent The three proofs share a mechanism. A block that severs its inputs has
its own complement as cut set, hence damages at least one of its own units; if it has
two units it damages two; and if it is neither small nor almost everything it damages
three. Counting the blocks that can afford to be undamaged then forces the shape of the
partition, and the normalizer follows. We state the three lemmas without proof: each is a
case analysis of a size that makes a prose rendering unilluminating, and the proof is the
Lean development, whose corresponding theorem is named after each statement.

Against these three floors we set a single witness: the partition
$\vartheta^{\ast}$ of $\ZN$ into three arcs of sizes $a, k, k$ with $k = \lfloor (N+1)/3
\rfloor$ and $a = N - 2k$, which damages at most four units. A computation
(\texttt{sysCutCount\_threeArc}) shows that its normalizer is large enough to beat each
floor \emph{strictly}:
\[
  4(N-1) < \mathcal{N}(\vartheta^{\ast}), \qquad
  N^{2} < 2\,\mathcal{N}(\vartheta^{\ast}), \qquad
  (N-1)^{2} + 8(N-1) < 3\,\mathcal{N}(\vartheta^{\ast}),
\]
the three inequalities saying exactly that $\mathrm{dmg}(\vartheta^{\ast}) /
\mathcal{N}(\vartheta^{\ast})$ is smaller than the least ratio available at damage one,
at damage two and at damage three respectively. Hence no partition damaging three or
fewer units minimizes the ratio, so every minimizer damages at least four, and the
supremum over the minimizing set in Proposition~\ref{prop:damage} is at least four. This
is \texttt{four\_le\_sysPhi\_circNet}. Note where the weight falls: the conclusion is not
that the ratio is monotone in the damage, which is false, but that a single explicit
partition undercuts all three floors at once; and the passage from ``every minimizer
damages at least four'' to ``$\phis \ge 4$'' is exactly the supremum-over-argmin
convention of the definition, which is therefore load-bearing here.

\paragraph{Proper subsystems: $\phis < 4$.} Here the geometry of the cycle does the work,
and a deletion is worth more than it looks. A candidate $S$ with $|S| \le 1$ has
$\phis(S,s) = 0$, so assume $|S| \ge 2$ and choose a deleted unit $d \notin S$ whose
predecessor $d-1$ lies in $S$. Partition $S$ into $P = S \cap A$ and $S \setminus P$,
where $A$ is the arc of length $T$ running backwards from $d-1$, with $P$ severing its
inputs. A unit $j \in P$ is damaged only if a live input of $j$ lies outside $P$; but the
inputs of a unit of $A$ are its two successors along the arc, which remain in $P$ except
at the head, and the unit at the head reads $d$, which is deleted and therefore no longer
a live input at all. The deletion, in other words, does not merely cut a connection: it
removes one, so the head is not damaged either.

\begin{lemma}[Shielded witness]
\label{lem:shield}
For $N \ge 3$, $d \notin S$ with $d-1 \in S$, and $1 \le T \le N-2$ with $P$ and
$S \setminus P$ both nonempty, the partition above damages exactly one unit of $S$.
\textup{(\texttt{damagedCountOn\_shield})}
\end{lemma}

\begin{lemma}[Crude competitor bound]
\label{lem:crude}
For every proper $S \subsetneq \ZN$ with $|S| \ge 2$, every background and every state
$s$, $\phis(S,s) < 4$. \textup{(\texttt{sysPhi\_circNet\_lt\_four})}
\end{lemma}

\noindent Choosing $T$ so that $|P| = \lfloor |S|/2 \rfloor$ balances the two blocks, so
$\mathcal{N} = |P|\,(|S| - |P|) \ge (|S|^{2}-1)/4$ while the damage is $1$. Since
$\phis(S,s)$ is at most $|S|(|S|-1)$ times the minimal ratio, and that ratio is at most
$1/\mathcal{N}$,
\[
  \phis(S,s) \;\le\; \frac{|S|(|S|-1)}{(|S|^{2}-1)/4} \;=\; \frac{4|S|}{|S|+1} \;<\; 4 ,
\]
which is the bound required, with no case left over: $S$ need not be an arc, since only
the block $P$ was chosen to be one. The two halves together give
Theorem~\ref{thm:exclusion}. We emphasize that the theorem is proved for every state
and every background: by Proposition~\ref{prop:competitor} the values involved do not
depend on the state when the relevant submatrices are invertible, and when they are
singular the state can only decrease $\phis(S,s)$.

\section{The distinctions of the cycle}
\label{sec:arcs}

We now turn to the mechanism level. Call $M \subseteq \ZN$ an \emph{arc} when it is a set
of consecutive residues.

\begin{proposition}[Arcs are distinctions]
\label{prop:arcsare}
Let $N \ge 3$ with $3 \nmid N$ and let the state be all-ones. Every arc of length
$L \ge 3$ is a distinction of $\Rring$, the arc saturating to the whole cycle once
$L \ge N$. Moreover its maximally irreducible cause purview contains the arc itself, so
that the support of the distinction contains every unit of $M$ in its stated state.
\end{proposition}

The proof is the mechanism-level counterpart of the damage computation. Write $D(M)$ for
the set of units all three of whose inputs lie in $M$; for an arc of length $L$ this is
an arc of length $L-2$, nonempty exactly when $L \ge 3$, and the union of the windows of
$D(M)$ is $M$ again. That gives a positive effect side. For the cause side, an
exclusive-or unit constrains its whole window, so a purview over which the cause
information is positive must contain every window of every unit of $M$; in particular it
contains $M$. Positivity of $\phic$ then follows from severing a single unit, which costs
a factor $2^{-(N - |M|)}$, and $\phid = \min\{\phie,\phic\} > 0$.

The converse holds as well, and it is what fixes the growth rate.

\begin{proposition}[Only arcs]
\label{prop:onlyarcs}
The distinctions of $\Rring$ at the all-ones state are exactly the arcs of length
between $3$ and $N-1$ together with the whole cycle; there are $N(N-3)+1$ of them.
\end{proposition}

\begin{remark}[Status]
\label{rem:status}
Proposition~\ref{prop:onlyarcs} is the one statement in this paper that is \emph{not}
part of the formal development. It is established by exhaustive evaluation over all
$2^N$ mechanisms for $N = 8, 10, 11, 13$, cross-checked against an independent exact
implementation, with the counts $41, 71, 89, 131$ agreeing with $N(N-3)+1$ and with no
non-arc distinction and no arc of length one or two appearing. The argument for general
$N$ is the criterion used in that evaluation---the effect information of $M$ is positive
if and only if some connected component of $D(M)$ in the window-overlap graph has window
union equal to $M$, and the cause information is positive if and only if $M$ is connected
in that graph---from which the classification follows; formalizing it is routine but was
not carried out. Nothing in Theorems~\ref{thm:exclusion} or~\ref{thm:lower} depends on
it. Its role here is to show that the quadratic rate is forced: a family of $2^{cN}$
distinctions sharing a unit would give $\Phi = 2^{2^{\Theta(N)}}$, and
Proposition~\ref{prop:onlyarcs} says no such family exists on this cycle.
\end{remark}

\section{The lower bound}
\label{sec:lower}

The mechanism that converts many distinctions into large $\Phi$ is the following, and it
is where the relations enter.

\begin{lemma}[Family bound]
\label{lem:family}
Let $F$ be a family of distinctions of a substrate, all containing a common \emph{stated
unit} $p$, that is a pair (unit, state), in their supports, each satisfying $\phid(d)/|\mathrm{supp}(d)| \ge c > 0$. Then
\[
  \Phi_{\max} \;\ge\; \bigl(2^{|F|} - 1 - |F|\bigr)\,c .
\]
\end{lemma}

This is not new, and we claim no credit for it: it is the counting identity
of~\cite[Eq.~11]{zaeemzadeh2024} specialized to a single stated unit. That identity
resolves the relation sum as
\[
  \sum_{\mathbf{d} \subseteq \mathcal{D}} \phir(\mathbf{d})
  \;=\; \sum_{|\mathbf{d}| = 1} \phir(\mathbf{d})
     \;+\; \sum_{o} \sum_{i=1}^{|\mathcal{Z}(o)|}
        \frac{\phid^{(i)}}{|z^{(i)}|}\bigl(2^{|\mathcal{Z}(o)| - i} - 1\bigr),
\]
where $\mathcal{Z}(o)$ collects the distinctions whose support contains the stated unit
$o$. Taking a single $o$, discarding the self-relations, and putting every ratio at its
floor $c$ leaves $c \sum_{i=1}^{|F|} (2^{|F| - i} - 1) = \bigl(2^{|F|} - 1 - |F|\bigr) c$,
which is the lemma. Concretely: every subset of $F$ with at least two members is a
relation, since its distinctions overlap congruently at $p$, and each contributes at least
$c$. Only the cardinality of the family and the uniform floor enter; the substrate is
arbitrary. What the present paper supplies is not this step but its input --- a family of
certified size, Proposition~\ref{prop:arcsare} --- and that is where the work lies.

One point deserves comment, since the two natural readings of $\phir$ differ in general.
The formalization computes the overlap as the intersection $\bigwedge_{D \in \mathbf{d}}
\mathrm{supp}(D)$ of the supports, whereas~\cite[Eq.~55]{iit40} takes the union of the
face purviews of~\cite[Eq.~50]{iit40}. Each face purview is contained in that
intersection, and conversely every stated unit of the intersection lies on some face as
soon as $|\mathbf{d}| \ge 2$; the two therefore agree for relations, which is the only
case in which $\phir$ is used. The equality of the two is a theorem of the imported
library~\cite{integ-repo} (\texttt{mem\_inf\_support\_iff\_exists\_face}), so the
simplified formula is not an approximation but the same quantity, and no discrepancy can
arise here in either direction.

\begin{proof}[Proof of Theorem~\ref{thm:lower}]
Fix the unit $0$ and index arcs through it by a pair $(s,e)$ with
$1 \le s < \lfloor N/2 \rfloor$ and $1 \le e < N - \lfloor N/2 \rfloor$, letting
$F$ consist of the arcs running from $-s$ to $e$. Each contains $0$ and has length
$s + e + 1 \ge 3$, so by Proposition~\ref{prop:arcsare} each is a distinction whose
support contains the stated unit $(0, \mathrm{true})$. Distinct pairs give distinct arcs,
so
\[
  |F| \;=\; \bigl(\lfloor N/2 \rfloor - 1\bigr)
            \bigl(N - \lfloor N/2 \rfloor - 1\bigr) \;\ge\; N^{2}/8 .
\]
The support of each contains at most $2N$ stated
units and
$\phid \ge (1/2)^{N}$ by the cause-side estimate of Proposition~\ref{prop:arcsare}, so
$c = (1/2)^N/(2N)$ is a uniform floor. Lemma~\ref{lem:family} gives
\[
  \Phi_{\max} \;\ge\; \bigl(2^{|F|} - 1 - |F|\bigr)\cdot\frac{(1/2)^N}{2N}
  \;\ge\; 2^{\lfloor N^2/8\rfloor}\cdot\frac{(1/2)^N}{2N} - 1 . \qedhere
\]
\end{proof}

\begin{remark}[On the constants]
\label{rem:constants}
Both constants are generous. The true family has $\sum_{L=3}^{N-1}\min(L+2,N)+1 \approx
N^2/2$ members, measured as $35, 56, 68, 95$ at $N = 8, 10, 11, 13$. The true floor is
not $2^{-N}/(2N)$ but approximately $1/(2N)$: the support of an arc of length $L$ has
$\min(L+2,N)$ units, and $\phid$ equals $1/4$ for $L = 3$, $1/2$ for
$4 \le L \le N-2$, $3/2$ for $L = N-1$ and $4$ for the whole cycle. Substituting these
gives a bound of the form $2^{N^2/2}/(2N)$, non-vacuous from small $N$ on. We proved the
crude version because it suffices for the growth rate and because the sharp values
require the exact purview computation of Proposition~\ref{prop:onlyarcs}, which is not
formalized.
\end{remark}

\section{Method: direction, assistance, and certification}
\label{sec:method}

The development behind this paper was produced by a division of labour that we think is
worth describing, because the result would not have been reachable by either party alone
and because the arrangement generalizes.

The mathematical direction was human, beginning with the decision that set everything
else in motion: to study integrated information with a proof assistant at all. That is
not an obvious move. Integrated information theory is a theory of consciousness, its
literature is computational and empirical, and its central quantities had never been
formalized; nothing in the field's practice suggested that a question about them would
yield to interactive theorem proving rather than to more computing power. The decision to
formalize the definitions from first principles, and then to treat the formal system as
an instrument of \emph{discovery} rather than as a means of certifying results obtained
elsewhere, is the reason the question could be posed in a form admitting an answer at
all. From it followed the rest of the direction, which consisted largely in asking the
right questions at the right moments.

The assistance was mechanical and operated at two scales. Exact evaluation of small
systems---our own implementations, validated against a literal transcription of the
formal definitions, and against PyPhi---supplied the conjectures: the shape of the
extremal partitions, the value of $\phis$ on the whole and on competitors, the
identification of the distinctions as arcs. The formalization itself was then carried out
by language-model assistants, several of which were used,
working file by file against stated blueprints, with the statements fixed in advance by
us and the proofs left to them. They did not only execute: the criteria for positivity of
the effect and cause information, and more than one counterexample to a lemma we had
stated, came from them. Which assistant produced which proof is of no
consequence to the result: what is checked is the development, not its provenance.

Certification was the proof assistant's. This matters more than it may appear. In the
course of this work, several plausible intermediate claims turned out to be false, and
they were caught not by review but by the impossibility of closing a goal: a lemma about
severed blocks that holds for one wiring and fails for another; a trap lemma stated
without a damage budget, refuted by an eight-unit counterexample; and, most consequential
of all, the doubly-exponential conjecture itself. We also record a methodological hazard
of the computational side, which cost us a false landscape before it was noticed: exact
enumerations truncated by a damage bound return \emph{lower} bounds, and can hide the
true minimizer when it lies at high damage. Where a number in this paper comes from
enumeration, it comes from a certified or exhaustive run.

\section{The formal development}
\label{sec:formal}

The development is written in Lean~4 against Mathlib, and is in two parts.

The definitions of IIT~4.0 are formalized from first principles---repertoires, causal
marginalization, the integrated information of mechanisms, distinctions, relations,
$\Phi$, system partitions, $\phis$, and the exclusion postulate---in a separate library,
released as~\cite{integ-repo} and described in the companion
paper~\cite{mayeux-formalization}. That library transcribes the definitions and develops general
consequences of them, but proves nothing about any particular substrate.

The results of this paper form a second library~\cite{circ-repo}, and the relation between the two is
worth stating exactly, because it is what rules out a class of doubt that formalized
results are otherwise open to. The second library does not restate, re-transcribe, copy
or adapt the definitions: it \emph{imports} the first as a package dependency pinned to
release \texttt{v1}, and the build fails if that dependency is absent or altered.
Consequently the symbols $\phid$, $\phir$, $\Phi$, $\phis$ and the predicate defining a
complex, as they occur in Theorems~\ref{thm:exclusion} and~\ref{thm:lower}, are literally
the same Lean constants as in~\cite{integ-repo}---the same elaborated terms, not
equivalent renderings of them. No definition of that library is redeclared, overridden,
shadowed or specialized here, and nothing in this paper's development could weaken one
without failing to compile. It follows that whatever scrutiny the correspondence between
IIT~4.0 as published and IIT~4.0 as formalized has received, or will receive, in
connection with~\cite{mayeux-formalization} transfers to the present theorems unchanged;
there is no gap between the two formalizations in which a discrepancy could hide, because
there is only one formalization. What the second library adds is substrate-specific
material and nothing else. Theorem~\ref{thm:exclusion} is
\texttt{isComplex\_circNet\_univ} and Theorem~\ref{thm:lower} is
\texttt{two\_pow\_sq\_le\_PhiMax\_circNet}. Both, and every lemma they depend on, are
checked by the Lean kernel and rely on no axioms beyond the three of the ambient logic:
propositional extensionality, the axiom of choice, and quotient soundness. There are no
uses of \texttt{sorry}, of \texttt{native\_decide}, or of any other escape hatch, and no
appeal to numerical computation that the kernel does not itself perform. The two theorems
are checked in a single import closure, so they are established of the same object.

For scale: the imported definitions library is some $7{,}300$ lines, and the
substrate-specific one about $9{,}200$, of which roughly $7{,}200$ prove exclusion and
$2{,}000$ the lower bound---maximality is by some margin the more expensive half.

For concreteness we give the three principal statements verbatim, as they appear in the
development. Theorem~\ref{thm:exclusion}, exclusion, quantified over every background
\texttt{u} and every state \texttt{s}:

\begin{lean}
theorem isComplex_circNet_univ (h3 : ¬ (3 ∣ N)) (hN : 8 ≤ N) (u s : State N) :
    IsComplex (circNet N) u Finset.univ Finset.univ s
\end{lean}

\noindent Proposition~\ref{prop:arcsare}, that every arc of length at least three is a
distinction:

\begin{lean}
theorem isDistinctionMech_arc (hN : 3 ≤ N) (h3 : ¬ (3 ∣ N)) (u : State N) (a : Fin N)
    {L : ℕ} (hL : 3 ≤ L) :
    IsDistinctionMech (circNet N) Finset.univ u allOnes (arcAt a L)
\end{lean}

\noindent And Theorem~\ref{thm:lower}, the lower bound, at the all-ones state:

\begin{lean}
theorem two_pow_sq_le_PhiMax_circNet (h3 : ¬ (3 ∣ N)) (hN : 8 ≤ N) (u : State N) :
    (2 : ℝ) ^ (N ^ 2 / 8) * ((1 / 2 : ℝ) ^ N / (2 * N)) - 1
      ≤ PhiMax (circNet N) Finset.univ u allOnes
\end{lean}

\noindent Here \texttt{IsComplex}, \texttt{IsDistinctionMech} and \texttt{PhiMax} are the
constants of~\cite{integ-repo}, imported and not redefined; \texttt{circNet},
\texttt{arcAt} and \texttt{allOnes} are the cycle of Definition~\ref{def:ring}, its arcs,
and the all-ones state. Note that the two theorems quantify differently over the state,
as discussed in Section~\ref{sec:substrate}: \texttt{s} is universally quantified in the
first, whereas the second is stated at \texttt{allOnes}.

\subsection*{Acknowledgements}

This work rests on two communal efforts. The first is Mathlib, the mathematical library
of Lean~4, without which a development of this size would not have been attempted. The
second is PyPhi~\cite{pyphi}, against which the numerical work reported here was checked
and from which the conjectures that became Theorems~\ref{thm:exclusion}
and~\ref{thm:lower} originally came. The arrangement under which the formal proofs were
produced is described in Section~\ref{sec:method}.

\end{document}